\documentclass[12pt, reqno]{amsart}
\usepackage{amssymb}
\usepackage{amsmath}
\usepackage{graphicx,color}
\usepackage{wrapfig,framed}
\usepackage{mathtools}
\usepackage[height=22cm, width=15.7cm, hmarginratio={1:1}]{geometry}
\usepackage{hyperref}
\usepackage{float}

\theoremstyle{plain}
\newtheorem{theorem}{Theorem}
\newtheorem{prop}{Proposition}
\newtheorem{lemma}{Lemma}
\newtheorem{cor}{Corollary}
\newtheorem{conj}{Conjecture}
\theoremstyle{definition}

\newtheorem{remark}{Remark}

\newcommand{\beq}{\begin{equation}}
\newcommand{\eeq}{\end{equation}}
\newcommand{\nn}{\nonumber}

\newcommand{\Q}{\mathcal{Q}}
\newcommand{\C}{\mathcal{C}}

\newcommand{\QQ}{{\mathbb Q}}
\newcommand{\CC}{{\mathbb C}}

\newcommand{\M}{{\mathcal M}}
\newcommand{\bt}{{\bf t}}
\newcommand{\bs}{{\bf s}}
\newcommand{\e}{\epsilon}
\newcommand{\tu}{\tilde{u}}

\def\H{\mathcal{H}}
\def\MV{{\rm Vol} \, \Q}
\newcommand{\p}{\partial}

\def\={\; = \;}
\def\+{\; + \;}
\def\:={\, := \, }

\def\m{\,-\,}

\def\vv{v} 
\def\HHod{\H}
\def\Hod{H}

\usepackage{array}
\newcolumntype{M}[1]{>{\centering\arraybackslash}m{#1}}
\newcolumntype{R}[1]{>{\raggedleft\arraybackslash}m{#1}}
\newcolumntype{N}{@{}m{0pt}@{}} 

\begin{document}
\title[Masur--Veech volumes]{Masur--Veech volumes of quadratic differentials and their asymptotics}
\author{Di Yang,  Don Zagier, Youjin Zhang}
\dedicatory{Dedicated to the memory of Boris Anatol'evich Dubrovin}
\maketitle
\begin{abstract}
Based on the Chen--M\"oller--Sauvaget formula, we apply the theory of integrable systems 
to derive three equations for the generating series of the 
Masur--Veech volumes ${\rm Vol} \, \Q_{g,n}$ associated with the principal strata 
of the moduli spaces of quadratic differentials, 
and propose refinements of the 
conjectural formulas given in~\cite{DGZZ,ADGZZ}
 for the large genus asymptotics
 of ${\rm Vol} \, \Q_{g,n}$ and of the associated area Siegel--Veech constants.
\end{abstract}

\medskip

\setcounter{tocdepth}{1}
\tableofcontents

\section{Statements of the results}
Let $\M_{g,n}$ denote the moduli space of complex 
algebraic curves of genus~$g$ with~$n$ distinct marked points, 
and $\Q_{g,n}$ the moduli space of pairs $(\C,q)$, 
where $\C \in \M_{g,n}$ is a smooth algebraic curve and $q$ is a 
meromorphic quadratic differential on~$\C$ with only simple poles at the marked points. 
This moduli space of quadratic differentials~$\Q_{g,n}$ 
is endowed with the canonical symplectic structure. 
The induced volume element on $\Q_{g,n}$ is called 
the Masur--Veech (MV) volume element. 
Denote by~$\MV_{g,n}$ the volume of~$\Q_{g,n}$; see e.g.~\cite{DGZZ} for its meaning.
Recently,  Chen--M\"oller--Sauvaget~\cite{CMS} proved that  
the volumes ${\rm Vol} \, \Q_{g,n}$ with $2g-2+n>0$ 
can be expressed in terms of linear Hodge integrals as follows:
\beq\label{CMSformula}
{\rm Vol} \, \Q_{g,n} \= 2^{2g+1} \frac{\pi^{6g-6+2n} (4g-4+n)!}{(6g-7+2n)!} \,
\sum_{j=0}^g \int_{\overline{\mathcal{M}}_{g,3g-3+2n-j}}  
\frac{\lambda_j  \psi_{n+1}^2 \cdots \psi_{3g-3+2n-j}^2}{(3g-3+n-j)!} \,,
\eeq
where $\overline{\mathcal{M}}_{g,k}$ denotes the Deligne--Mumford 
compactification of~$\mathcal{M}_{g,k}$,  
$\psi_i$ denotes the first Chern class of the $i_{\rm th}$ tautological line 
bundle on $\overline{\mathcal{M}}_{g,k}$, and $\lambda_j$ denotes the 
$j_{\rm th}$ Chern class of the rank~$g$ Hodge bundle $\mathbb{E}_{g,k}$ 
on $\overline{\mathcal{M}}_{g,k}$.  
The goal of the present paper is to study the numbers 
${\rm Vol} \, \Q_{g,n}$ by using the Chen--M\"oller--Sauvaget (CMS) formula.

For $g,n\geq 0$, we define 
\beq\label{definisionagn}
a_{g,n} \= \begin{cases}
\sum_{j=0}^g \frac1{(3g-3+n-j)!}\int_{\overline{\mathcal{M}}_{g,3g-3+2n-j}} 
\lambda_j  \psi_{n+1}^2 \cdots \psi_{3g-3+2n-j}^2 \,, & 2g-2+n>0\,, \\
0\,, & {\rm otherwise.}\\
\end{cases}
\eeq
Note that the $a_{g,n}$ are rational numbers, and differ from ${\rm Vol} \, \Q_{g,n}$ only by 
some simple factors. Define a generating series $\H(x,\e)$ for the numbers $a_{g,n}$, called the MV free energy, by 
\beq\label{defintionofHxe}
\H(x,\e) \:= \sum_{g,n\geq0} \e^{2g-2} \frac{x^n}{n!} a_{g,n} \,.
\eeq
The first result of this paper is then given by the following theorem.
\begin{theorem} \label{thmequation}
The series $\H(x,\e)$ satisfies the following two equations:
\begin{align}
& \bigl[\p_x(\H_+\,-\, \H_-)\bigr]^2 \+ \p_x^2 \bigl( \H_+ + \H_-\bigr) \= \frac{2x}{\e^2} \,,  \label{theequation} \\
& \biggl(\e  \p_\e + \frac12x\p_x - \frac{\e^2}{24} \p_x^3\biggr)  \bigl(\H_{+} - \H_{-}\bigr)    \+ 
\frac{\e^2}{12}  \bigl[\p_x(\H_{+}-\H_{-})\bigr]^3   \=  0  \,, \label{anotherequation} 
\end{align} 
where 
$\H_{\pm} := \H\bigl(x\pm\frac{i \e}{2},\e\bigr)$.
\end{theorem} 

A statement equivalent to equation~\eqref{theequation} is given by the following corollary. 
\begin{cor}\label{thmrecursion}
For all $g\ge0$ and $n\geq 2$, the numbers $a_{g,n}$ 
can be uniquely determined by the following recursion relation
\begin{align}
&  a_{g,q+2} \= \frac{q!}2 \sum_{g_1+g_2+j_1+j_2=g \atop n_1+n_2=q+4+2(j_1+j_2)} 
\frac{(-1)^{j_1+j_2}  a_{g_1,n_1} a_{g_2,n_2} }{4^{j_1+j_2} (2j_1+1)!(2j_2+1)! (n_1-2j_1-2)!(n_2-2j_2-2)!}  \nn\\
& \qquad \qquad \,-\, \sum_{j=1}^g \frac{(-1)^j a_{g-j,q+2j+2}}{4^j (2j)!}   \+ \delta_{q,1} \delta_{g,0}   
 \label{recursionagn}
\end{align}
along with the boundary condition $a_{0,2}=0$ (cf.~\eqref{definisionagn}), where $q\geq 0$.
\end{cor}

Another corollary of Theorem~\ref{thmequation} 
is the following non-linear differential equation for the series~$\H$.
\begin{cor}\label{Kazdeduced} 
The series $\H=\H(x,\e)$ satisfies the following equation:
\beq \label{deducedfrom}
 \e \p_\e \p_x (\H) \+ x  \p_x^2 ( \H) \+ \frac12  \p_x(\H) \,-\, \frac{\e^2}4  \bigl[\p_x^2(\H)\bigr]^2 \,-\, \frac{\e^2}{24}  \p_x^4(\H)  \= 0 \,.
\eeq
\end{cor}
\noindent The proof will be given in Section~\ref{section3}. 
We also show there that equation~\eqref{deducedfrom} implies a recursion given by Kazarian in~\cite{Kazarian} for the Hodge integrals   
\[\frac{(5g-3-j)(5g-5-j)}{(3g-3-j)!}\int_{\overline{\mathcal{M}}_{g,3g-3-j}} 
\lambda_j  \psi_{1}^2 \cdots \psi_{3g-3-j}^2\,,\quad 0\le j\le g \, .\]

A third corollary of Theorem~\ref{thmequation} (which apart from the boundary conditions 
is in fact equivalent to equation~\eqref{deducedfrom}) is the following recursion for the numbers~$a_{g,n}$. 
\begin{cor}\label{thmrecursion}
For all $g\geq 0$ and $n\geq 1$, the numbers $a_{g,n}$ are given recursively by
\begin{align}
&  a_{g, n} \= 
\frac{1}{2} \sum_{g_1,g_2\ge 0 \atop g_1+g_2=g} \sum_{ n_i\geq 2,  (g_i,n_i)\neq (0,3), i=1,2 \atop n_1+n_2=n+3}  
\binom{n-1}{n_1-2}
\, \frac{a_{g_1,n_1} a_{g_2,n_2}}{4g-4+n}   \+  \frac1{12} \, \frac{a_{g-1,n+3}}{4g-4+n}  
 \label{recursionagntwo} 
\end{align}
if $2g-2+n > 0$, $(g,n) \not\in \{ (0,3) , \,  (0,4) \}$, $a_{0,3}=a_{0,4}=1$ and $a_{0,1}=a_{0,2}=a_{-1,n}=0$.
\end{cor}
The recursion relations~\eqref{recursionagn} or~\eqref{recursionagntwo} both give rapid (polynomial-time) 
algorithms for computing $a_{g,n}$ for $n\geq 2$ or $n\geq 1$, respectively. The first few values $a_{g,n}$ 
are given by the following table. 
\begin{table}[H]
\begin{center}\tiny 
    \begin{tabular} {|c|c|c|c|c|c|c|c |N}
    \hline
    & $n=0$ & $n=1$ & $n=2$ & $n=3$ & $n=4$ & $n=5$ & $n=6$ &   \\ [8pt]
    \hline
    $g=0$ &  0 & 0 & 0 & 1 & 1 &  3  & 15  & \\ [8pt]
    \hline
    $g=1$ & 0 & $\frac1{12}$ & $\frac{1}{8}$ & $\frac{11}{24}$ & $\frac{21}{8}$ & $\frac{163}{8}$ & $\frac{1595}{8}$  &  \\[8pt]
    \hline
    $g=2$ & $\frac1{96}$ & $\frac{29}{640}$ & $\frac{337}{1152}$ & $\frac{319}{128}$ & $\frac{10109}{384}$ & $\frac{42445}{128}$ & 
    $\frac{620641}{128}$  & \\[8pt]
    \hline
    $g=3$ & $\frac{575}{21504}$ & $\frac{20555}{82944} $ & $\frac{77633}{27648}$ & 
    $\frac{1038595}{27648}$ & $\frac{16011391}{27648}$ & $\frac{31040465}{3072}$ & 
    $\frac{201498115}{1024}$ &  \\[8pt]
    \hline
    $g=4$ & $\frac{2106241}{7962624}$ & $\frac{1103729}{294912}$ & $\frac{160909109}{2654208}$ & 
    $\frac{14674841399}{13271040}$ & $\frac{99177888029}{4423680}$ &
    $\frac{442442475179}{884736}$ & $\frac{10765584400823}{884736}$ &  \\[8pt]
  \hline
    \end{tabular}
\end{center}
\caption{The numbers $a_{g,n}$ with $0\leq g\leq 4$ and $0\leq n\leq 6$.} \label{table1}
\end{table}

The following proposition describes the property of ${\rm Vol} \, \Q_{g,n}$,
which will enable us to determine also $a_{g,0}$ and $a_{g,1}$ from~\eqref{theequation}, 
and $a_{g,0}$ from~\eqref{anotherequation} or~\eqref{deducedfrom}.
\begin{prop}[\cite{AEZ,ABCDGLW,CMS}] \label{prp1lem}
The following properties of the MV volumes hold:
\begin{align}
&{\rm Vol} \, \Q_{0,n}
\= \frac{\pi^{2n-6}}{2^{n-5}}\,,\qquad \forall \, n\ge 3\,;\label{thm0n}\\
&{\rm Vol} \, \Q_{1,n}
\= \frac{\pi^{2n}}3 \biggl(\frac{n! }{ (2 n-1)!!} \+ \frac{2n}{(2 n-1) 2^n}\biggr)\,, \qquad \forall \,n\ge 1\,;\label{thm1n}\\
&{\rm Vol} \, \Q_{g,n} \= 2^{2g+1+n} \frac{\pi^{6g-6+2n} (4g-4+n)!}{(6g-7+2n)!} 
\sum_{j=0}^{g}  \frac{\langle \lambda_j \tau_2^{3g-3-j} \rangle_g}{(3g-3-j)!} \biggl(\frac{5g-5-j}{2}\biggr)_n\,,\label{thmgn}
\end{align}
where $g\geq 2$, $n\geq 0$, $(b)_n:=b(b+1)\cdots (b+n-1)$ denotes the increasing Pochhammer symbol, and 
we used Witten's notation: for a cohomology class $\gamma\in H^*(\overline{\M}_{g,n};\CC)$, 
\[\langle \gamma \tau_{i_1} \cdots \tau_{i_n} \rangle_g \:=  
\int_{\overline{\mathcal M}_{g,n}} 
\gamma \, \psi_1^{i_1} \cdots \psi_n^{i_n} \,,\quad i_1,\dots,i_n\geq 0 \,. \]
\end{prop}
\noindent The explicit expression for ${\rm Vol} \, \Q_{0,n}$, $n\geq 3$ was conjectured 
by Kontsevich, and was proved by Athreya-Eskin-Zorich in~\cite{AEZ}. 
The formula~\eqref{thm1n} was conjecturally given by Andersen et. al.~\cite{ABCDGLW}, and
the formula~\eqref{thmgn} is equivalent to the 
Conjecture~5.4 of~\cite{ABCDGLW} (to see the equivalence, cf.~\cite{CMS}). 
A proof of Proposition~\ref{prp1lem} was given in~\cite{CMS}. 
In this paper we give a different proof of this proposition based on the following lemma. 

\begin{lemma} \label{Txstructure} 
Let $T=\sqrt{1-2x}$. Define 
the power series $\H_g(x)$, $g\geq 0$ by   
\beq\label{Hgenusexpansion}
\H(x,\e) \;=:\; \sum_{g\geq0} \e^{2g-2} \H_g(x) \,.
\eeq
Then we have
\begin{align}
\HHod_0(x) &\= 
\frac1{40}  \,-\, \frac{T^2}{12} \+ \frac{T^4}{8} \,-\, \frac{T^5}{15}\,, \label{h0hodgeT} \\
\HHod_1(x) & \= \frac1{24} \log \frac1T \+ \frac{1}{24}  \,(1-T) \,, \label{h1hodgeT} \\
\HHod_2(x) & \= \frac{7}{1440} \, \frac{1}{T^5} \+ \frac{5}{1152} \, \frac{1}{T^4} \+ \frac{7}{5760} \, \frac{1}{T^3}  \,. \label{h2hodgeT} 
\end{align}
In general, we have the following expression for $\HHod_g(x)$:
\begin{align}
\HHod_g(x) & \= 
\sum_{j=0}^{g}  \frac{\langle \lambda_j \tau_2^{3g-3-j} \rangle_g}{(3g-3-j)!} \frac{1}{T^{5g-5-j}}\,, \quad g\geq 2\,. \label{ggeq2hodgeT}
\end{align}
\end{lemma}
\noindent We give in Section~\ref{section2} a proof of Lemma~\ref{Txstructure} 
by using the CMS formula~\eqref{CMSformula} and the Dubrovin-Zhang formalism~\cite{DZ-norm, DLYZ1, DuY} on  
Hodge integrals.
Substituting the expansion~\eqref{Hgenusexpansion} into~\eqref{deducedfrom} we find
\beq\label{hgrecursion6}
x \H_{g}'' \+ \Bigl(2g-\frac32\Bigr) \H_{g}' \,-\, 
\frac14 \sum_{g_1,g_2\geq 0 \atop g_1+g_2=g} \H_{g_1}'' \H_{g_2}'' \,-\, \frac1{24} \H_{g-1}'''' \=0 \,.
\eeq
Here, prime, ``\,$'$\," denotes $d/dx$. It turns out that this formula together with Lemma~\ref{Txstructure} 
determines $\H_g$, $g\geq 0$, and therefore the $a_{g,n}$, uniquely for all $g,n\geq0$.

Recently, Aggarwal, Delecroix, Goujard, Zograf and Zorich~\cite{ADGZZ} 
proposed a conjectural formula for the large~$g$ leading asymptotics 
of ${\rm Vol} \, \Q_{g,n}$ (the conjectural formula was given originally in~\cite{DGZZ} for $n=0$). 
The ADGZZ conjecture was very recently proved in~\cite{A3}.
Our next result is a refinement of the ADGZZ conjecture 
to the following more precise asymptotic statement.
\begin{conj}\label{conjasymg}
For any fixed $n\geq 0$, we have the asymptotic formula:
\beq
{\rm Vol} \, \Q_{g,n} \;\sim\; \frac{2^{12g+4n-10}}{3^{4g+n-4}\pi}  \, \sum_{k=0}^\infty \frac{m_k(n)}{g^k}\,,\quad g\to\infty\,,  \label{mvlargeg} 
\eeq
where each $m_k(n)$ is a polynomial in~$n$ with coefficients in $\QQ[\pi^2]$, the first four values 
(with $M=-\pi^2/144$  for convenience)
given by 
\begin{align}
m_0(n) &\= 1\,, \quad m_1 (n) \= M\,, \nn\\
 m_2 (n) & \= \frac M{24}\,n^3\,-\,\frac{3M}8\,n^2 \+\frac{4M-17M^2}{6}\,n \+\frac{M+19M^2}2\,,  \nn \\ 
 m_3 (n) &\= - \frac{8M+27M^2}{288}n^4 \+ \frac{17M+65M^2}{48} n^3 
  \,-\,\frac{860M+1890M^2-14256M^3}{576} n^2 \nn \\  
 &\qquad \+ \frac{104M \,-\, 373M^2 \,-\, 6156 M^3 }{48}  n 
 \nn\\
 &\qquad \, - \, \frac{ 55M - 3615M^2 - 28650M^3 + 126846M^4 }{180} \,.\nn 
\end{align}  
\end{conj}

The asymptotic formula~\eqref{mvlargeg} with $\sum_{k=0}^\infty m_k(n)/g^k$ 
replaced by~$1$ is the ADGZZ conjecture. We refer to~\cite{A1,A2,CMSZ,CMZ,EZ,S2,S3} for 
the analogues of the ADGZZ conjecture and Conjecture~\ref{conjasymg} (cf. also Conjecture~\ref{conjecture2} 
in Section~\ref{section4} below) for the MV volumes and for the related area Siegel--Veech constants associated with 
the moduli spaces of abelian differentials, and the proofs of these analogues via different approaches.
Conjecture~\ref{conjasymg} can also be stated in terms of the numbers $a_{g,n}$ defined in~\eqref{definisionagn} as
\beq a_{g,n} \;\sim\; \frac{(6g-7+2n)!}{(4g-4+n)! }\frac{2^{10g+4n-11}}{3^{4g+n-4}\pi^{6g-5+2n}} \, 
  \sum_{k=0}^\infty \frac{m_k(n)}{g^k} \,, \quad g\rightarrow \infty\,.
\eeq

Conjecture~\ref{conjasymg}, like the related Conjecture~\ref{conjecture2} which will be stated in Section~\ref{section4} below, 
is completely empirical. Specifically, we computed the 
values of $a_{g,n}$ numerically for $g\le 100$ and a number of small values of~$n$, then interpolated by the 
numerical method explained in~\cite{Zagier}, \cite[Section~5]{GM}
and elsewhere to get an asymptotic power series in~$1/g$ with coefficients 
known to high precision, and then used polynomial interpolation and the LLL (Lenstra--Lenstra--Lovasz) method to recognize the coefficients 
as polynomials in~$n$ with coefficients in $\QQ[\pi^2]$.

\begin{remark}
It would be interesting to investigate the following generating series:
\beq
C_n(\e)\:= \sum_{g\geq 0} \e^{2g-2} a_{g,n}\,, \quad n\geq 0\,.
\eeq
In other words, $\H(x,\e)=\sum_{n\geq 0}\frac{x^n}{n!} C_n(\e)$. 
Equation~\eqref{deducedfrom} then implies the following relations for $C_n(\e)$:
\beq
C_{n+4} \= \frac{24}\e C_{n+1}' \+ 12\frac{2n+1}{\e^2} C_{n+1} \,-\, 6n! \sum_{n_1+n_2=n} \frac{C_{n_1+2} C_{n_2+2}}{n_1! \,n_2!} \,, \quad n\geq 0\,. 
\eeq
Similarly, equation~\eqref{anotherequation} implies relations for the analogue of~$C_n(\e)$ for $\H_+-\H_-$.
Understanding of $C_n(\e)$ or its analogue might be useful for proving the above Conjecture~\ref{conjasymg}. 
\end{remark}

The paper is organized as follows. 
In Section~\ref{section2}, we review the Dubrovin-Zhang theory and give a proof of Lemma~\ref{Txstructure}.
In Section~\ref{section3}, we prove Theorem~\ref{thmequation}.
In Section~\ref{section4}, we extend a conjectural formula for the large genus asymptotics of the 
area Siegel--Veech constants.

\medskip

\noindent {\bf Acknowledgements.} 
We would like to thank Dawei Chen, Martin M\"oller, and Motohico Mulase
for helpful suggestions. Part of the work of D.Y. was done during his visit in MPIM; 
he thanks MPIM for excellent working conditions and financial supports. 
This work was partially supported by NSFC No.\,11771238.

\section{The Hodge free energy}\label{section2}
In this section we first give a short review of the Dubrovin-Zhang approach to Hodge integrals~\cite{DLYZ1,Buryak,DZ-norm,DuY}, 
and then specialize our discussions to linear Hodge integrals 
and prove Lemma~\ref{Txstructure}.
Recall that the genus~$g$ Hodge free energy $\HHod_g(\bt;\bs)$ is defined by
\begin{align}
&\HHod_g(\bt;\bs) \= \sum_{k\geq 0} \sum_{\,i_1, \dots, i_k\geq 0} \frac{t_{i_1}\cdots t_{i_k}}{k!} \int_{\overline{\mathcal M}_{g,k}} 
 \Omega_{g,k}(\bs) \, \psi_1^{i_1} \cdots \psi_k^{i_k}\,,\\
&\Omega_{g,k}(\bs) \:= \exp \biggl(\,\sum_{j\geq 0} s_{2j-1} \, {\rm ch}_{2j-1}(\mathbb{E}_{g,k})\biggr) \,.
\end{align}
Here $g\geq 0$, $\bt = (t_0,t_1,\dots)$, $\bs=(s_1,s_3,\dots)$, 
$t_0,t_1,t_2,\dots$, $s_1,s_3,\dots$ are indeterminates, 
and ${\rm ch}_1, {\rm ch}_3, {\rm ch}_5$, $\dots$ denote components of the Chern character of~$\mathbb{E}_{g,k}$. 
Define the total Hodge free energy $\HHod$ by
$$
 \HHod\=\HHod(\bt;\bs;\epsilon) \= \sum_{g\geq0} \HHod_g(\bt;\bs)\,\epsilon^{2g-2} \,.
$$

Let $\vv \in \CC[[\bt]]$ be the unique power series solution to the following equation:
\beq\label{hodo}
\sum_{i\geq 0} \frac{t_i}{i!} v^i \= v \,. 
\eeq
It is well known that this unique power series $\vv=\vv(\bt)$ has the explicit expression
\beq\label{genus0solution}
\vv(\bt) \=  
\sum_{k\geq 1} \frac1k \sum_{p_1,\dots,p_k\geq 0\atop p_1+\dots+p_k=k-1} \frac{t_{p_1}}{p_1!} \cdots \frac{t_{p_k}}{p_k!}\;.
\eeq
Denote 
\beq
v_m(\bt) \= \p_{t_0}^m(v(\bt))\,, \quad m\geq 0\,. 
\eeq

\medskip

\noindent {\bf Theorem A} (\cite{DLYZ1})  
\textit{The genus~0 and~1 Hodge free energies have the expressions
\begin{align}
& \HHod_0(\bt; \bs) \= 
\frac{\vv(\bt)^3}6 \m \sum_{i\geq 0} t_i  \frac{\vv(\bt)^{i+2}}{i! (i+2)}  
\+ \frac12 \sum_{i,j\geq 0} t_i t_j \frac{\vv(\bt)^{i+j+1}}{(i+j+1) i! j!}\,,\label{Hodgegenus0}\\
& \HHod_1(\bt; \bs) \= \frac1{24} \log\vv_1(\bt) \+ \frac{s_1}{24} \vv(\bt) \,.\label{Hodgegenus1}
\end{align}
For $g\geq 2,$ there exist elements 
\[\Hod_g(z_1,\dots,z_{3g-2}; s_1,s_3,\dots,s_{2g-1})\in \CC\bigl[z_1,\dots,z_{3g-2},z_1^{-1}; s_1,s_3,\dots,s_{2g-1}\bigr]\]
satisfing the conditions
\begin{align} 
& \sum_{m=1}^{3g-2} m \, z_m \, \frac{\p \Hod_g}{\p z_m} \= (2g-2)\,\Hod_g \, , \label{homog1}\\
& \sum_{m=2}^{3g-2} (m-1) \, z_m \, \frac{\p \Hod_g}{\p z_m} \+ \sum_{j=1}^g (2j-1) \, s_{2j-1} \, \frac{\p \Hod_g}{\p s_{2j-1}} 
  \= (3g-3)\, \Hod_g \, , \label{homog2}
\end{align}
such that 
\begin{equation}
\HHod_g(\bt; \bs) \= \Hod_g \bigl( \vv_1(\bt),\dots,\vv_{3g-2}(\bt);s_1,s_3,\dots,s_{2g-1}\bigr) \,. \label{generalhodge}
\end{equation}}

\noindent This theorem was proved in~\cite{DLYZ1}; see also~\cite{DuY} for a straightforward proof.

Define 
\beq\label{defineu}
u \= u(\bt;\bs;\e) \:= \e^2 \frac{\p^2 \H(\bt;\bs;\e)}{\p t_0^2} \,,
\eeq
then according to~\cite{DLYZ1}, $u$ satisfies an integrable hierarchy of tau-symmetric Hamiltonian evolutionary 
PDEs, called the Hodge hierarchy,  
which is a deformation of the KdV hierarchy~\cite{Witten,Kontsevich} and has the form
\beq
\frac{\p u}{\p t_k} \= P \frac{\delta \bar h_k}{\delta u(x)} \,, \quad k\geq 0.
\eeq
Here $P=\p_x \+ \cdots$ is a Hamiltonian operator, $\bar h_k$, $k\geq 0$ are Hamiltonians.

In~\cite{DYZ} Theorem~A was applied under a particular specialization of $\bt,\bs$, which  
gives the classical Hurwitz numbers according to the ELSV formula. In this paper, we consider a {\it different} specialization.
Firstly, we specialize~${\bf s}$ to ${\bf s}={\bf s}^*$ as follows:
\beq
s_{2k-1}^* \:= (2k-2)! \, s^{2k-1}, \quad k\geq 1 \,.
\eeq
Denote by $\Lambda_{g,k}(s):=\sum_{j=0}^g \lambda_j \, s^j$ the Chern polynomial of~$\mathbb{E}_{g,k}$. Applying the relationship between 
the Chern classes and the Chern character, and 
using Mumford's relations~\cite{Mumford} \[{\rm ch}_{2m}(\mathbb{E}_{g,k})\=0\,, \quad m\geq 1\,,\] 
we obtain $\Omega_{g,k}(\bs=\bs^*) =  \Lambda_{g,k}(s)$. So we have 
\beq\label{linearhodgepotential}
\HHod_g(\bt;\bs^*) \= \sum_{n\geq 0} \sum_{\,i_1, \dots, i_n\geq 0} \frac{t_{i_1}\cdots t_{i_n}}{n!} \int_{\overline{\mathcal M}_{g,n}} 
 \Lambda_{g,n}(s) \, \psi_1^{i_1} \cdots \psi_n^{i_n} \,.
\eeq
Secondly, we specialize $\bt$ to $\bt=\bt^*$ given by
\beq\label{specialvalue}
t_0^* \= x \,, \quad t_1^*  \= 0 \,,  \quad t_2^* \= 1\,, \quad t_i^* \= 0~(i\geq 3) \,.
\eeq
Substituting~\eqref{specialvalue} into~\eqref{linearhodgepotential} we arrive at
\beq\label{Hodgextj}
\HHod_g(\bt^*;\bs^*) \=   \sum_{n_0\geq0} 
\frac{x^{n_0}}{n_0!} \, \sum_{j=0}^g s^j
 \frac{\langle \lambda_j  \, \tau_0^{n_0} \, \tau_2^{3g-3+n_0-j} \rangle_g}{ (3g-3+n_0-j)!} \,.
\eeq
From the definition of~$a_{g,n}$ given in~\eqref{definisionagn}, it follows that 
the MV free energy is a specialized linear Hodge free energy. More precisely, we have the following lemma.
\begin{lemma} \label{twofreeenergiesequal}
For any $g\geq 0$, the following identities hold:
\beq
\H_g(x) \= \HHod_g(\bt^*;\bs^*)|_{s=1} \,,
\eeq
where $\H_g(x)$ is the $g^{\rm th}$ part of the MV free energy~\eqref{Hgenusexpansion}.
Equivalently, we have
\beq\label{MVhodge}
{\rm Vol} \, \Q_{g,n} 
\= 2^{2g+1} \frac{\pi^{6g-6+2n} (4g-4+n)!}{(6g-7+2n)!} \,\p_x^n\bigl(\HHod_g(\bt^*;\bs^*)\bigr)\big|_{x=0,s=1} \,.
\eeq

\end{lemma}

Let us now apply Theorem~A to the computation of~$\HHod_g(\bt^*;\bs^*)$, which, due to~\eqref{MVhodge}, 
gives rise to~${\rm Vol} \, \Q_{g,n}$.
Substituting~\eqref{specialvalue} into~\eqref{hodo} we find that $v=v(\bt^*)$ satisfies the following quadratic equation
\beq
x \+ \frac{v^2}2  \= v\,.
\eeq
By solving this and observing that the power series~$v$ starts with~$x$, 
we obtain 
\[v(\bt^*)= 1-\sqrt{1-2 x}\,.\]
Denote 
\beq
T \:= \sqrt{1-2x} \, .
\eeq
Then by noticing $ \p_x \= -\frac1T \p_T  $
we find 
\beq\label{vmam}
v_m (\bt^*) \=  \frac{(2m-3)!!}{T^{2m-1}} \+ \delta_{m,0} \,, \quad m\geq 0\,.
\eeq

\begin{lemma} \label{jetsstructure} 
The power series $\HHod_g(\bt^*;\bs^*)$ of $x,t$ 
are given explicitly for $g=0,1,2$ by
\begin{align}
\HHod_0(\bt^*;\bs^*) &\= 
\frac1{40}  \,-\, \frac{T^2}{12} \+ \frac{T^4}{8} \,-\, \frac{T^5}{15} \,, \label{h0hodges} \\
\HHod_1(\bt^*;\bs^*) & \= \frac1{24} \log \frac1T \+ \frac{s}{24}  \,(1-T) \,, \label{h1hodges} \\
\HHod_2(\bt^*;\bs^*) & \= \frac{7}{1440} \, \frac{1}{T^5} \+ \frac{5}{1152} \, \frac{s}{T^4} \+ \frac{7}{5760} \, \frac{s^2}{T^3}  \,. \label{h2hodges}
\end{align}
In general, for $g\geq 2$, $\HHod_g(\bt^*;\bs^*)$ has the following expression:
\beq\label{ggeq2hodges}
\HHod_g(\bt^*;\bs^*) \= 
\sum_{j=0}^{g}  \frac{\langle \lambda_j \tau_2^{3g-3-j} \rangle_g}{(3g-3-j)!} \frac{s^j}{T^{5g-5-j}}\,, \quad g\geq 2\,.
\eeq
\end{lemma}

\begin{proof}
By substituting~\eqref{vmam} into~\eqref{Hodgegenus0} and \eqref{Hodgegenus1}, we arrive at the
formulas for $\HHod_0(\bt^*;\bs^*)$ and $\HHod_1(\bt^*;\bs^*)$, respectively.
The formula for $\HHod_2(\bt^*;\bs^*)$ can be obtained by using the algorithm of~\cite{DLYZ1} with $v_m(\bt^*)$ given by~\eqref{vmam}.
To show the validity of the formula for $\HHod_g(\bt^*;\bs^*)$, $g\geq 2$, we first observe that, according to \eqref{generalhodge}, \eqref{vmam} 
and the homogeneity conditions \eqref{homog1}, \eqref{homog2}, the function $\HHod_g(\bt^*;\bs^*)$ 
can be written in the form
\beq\label{ggeq2hodgeform}
\HHod_g(\bt^*;\bs^*) \= \sum_{j=0}^{g}  \frac{C_{g,j} s^j}{T^{5g-5-j}} \,, \quad g\geq 2\,,
\eeq
where $C_{g,j}\in\QQ$.
Therefore,
\[
\HHod_g(\bt^*;\bs^*)|_{x=0} \= \sum_{j=0}^{g}  C_{g,j} s^j \, , \quad g\geq 2 \,.
\]
On the other hand, it follows from~\eqref{Hodgextj} that 
\[\HHod_g(\bt^*;\bs^*)|_{x=0} \= \sum_{j=0}^g  
\frac{\langle \lambda_j \tau_2^{3g-3-j} \rangle_g}{(3g-3-j)!} s^j \,.\]
By comparing the coefficients of $s^j$ in the two formulas given above we arrive at
\beq
C_{g,j} \= \frac{\langle \lambda_j \tau_2^{3g-3-j} \rangle_g}{(3g-3-j)!}\,,\quad j=0,\dots,g\,,
\eeq
where $g\ge 2$. The lemma is proved.
\end{proof}

\smallskip

\begin{proof}[Proof of Lemma~\ref{Txstructure}]
By putting $s=1$ in Lemma~\ref{jetsstructure}, we arrive at the result of Lemma~\ref{Txstructure}.
\end{proof}

Now let us give a proof of Proposition~\ref{prp1lem} based on Lemma~\ref{Txstructure}.

\begin{proof}[Proof of Proposition~\ref{prp1lem}]
By using~\eqref{h0hodgeT} and the fact that $\frac{d}{dx} = -\frac1T \frac{d}{dT}$ we have
\begin{align}
& \HHod_0'(x) \= \frac16 - \frac{T^2}2  + \frac{T^3}3 \,, 
\quad \HHod_0''(x) \= v(\bt^*) \,,\\
& \frac{d^n\HHod_0(x)}{dx^n}  \= v_{n-2}(\bt^*) \=  \frac{ (2n-7)!!}{T^{2n-5}} \,, \quad n\geq 3\,.
\end{align}
Therefore,
$
\frac{d^n \HHod_0(x) }{dx^n}\big|_{x=0} = 
(2n-7)!! \, \delta_{n\geq 3}$.
Due to the definition~\eqref{defintionofHxe} and the CMS formula this gives~\eqref{thm0n}. 
Similarly, by using~\eqref{h1hodgeT} we obtain
\beq
\frac{d^n\HHod_1(x)}{dx^n} \= 
\frac{\delta_{n\geq1}}{24} \frac{2^{n-1}(n-1)!}{T^{2n}} \+\frac{\delta_{n,0}}{24} \log \frac1T \+ \frac{1}{24} \frac{(2n-3)!! }{T^{2n-1}} \+ \frac{\delta_{n,0}}{24} \,,
\eeq
from which we arrive at~\eqref{thm1n}.
Finally, by using~\eqref{ggeq2hodgeT} we have for $g\geq 2$,
\beq
\frac{d^n\HHod_g(x)}{dx^n} \= 
\sum_{j=0}^{g}\frac{\langle \lambda_j \tau_2^{3g-3-j} \rangle_g}{(3g-3-j)!} \frac{\prod_{i=0}^{n-1}(5g-5-j+2i)}{T^{5g-5-j+2n}} \, ,
\eeq
which yields formula~\eqref{thmgn}. 
Proposition~\ref{prp1lem} is proved. 
\end{proof}

\begin{remark}
The explicit expressions of the numbers 
$\langle \lambda_g \tau_2^{2g-3} \rangle_g$ that appear in~\eqref{thmgn} of Proposition~\ref{prp1lem} are given by the following
$\lambda_g$-conjecture proven in~\cite{FP,ELSV}:
\beq
\frac{\langle \lambda_g \tau_2^{2g-3} \rangle_g}{(2g-3)!} \= 
\frac{2^{2g-1}-1}{2^{2g-1}} (4g-7)!! \, \frac{|B_{2g}|}{(2g)!} \,, \qquad g\geq 2\,,
\eeq
where $B_k$ denotes the $k_{\rm th}$-Bernoulli number. The number 
$\langle \tau_2^{3g-3} \rangle_g$ for $g\geq 2$
has the expression~\cite{IZ}:
\beq\label{IZformula}
\frac{\langle \tau_2^{3g-3} \rangle_g}{(3g-3)!} \= \frac{24^{-g} c_g}{(5g-3)(5g-5)} \,,
\eeq
where  $c_g$ are given by the recursion
\beq\label{cgdef}
c_g \=  50\,(g-1)^2\,c_{g-1} + \frac12 \sum_{h=2}^{g-2} c_h \, c_{g-h},\quad  g\geq 3
\eeq
together with 
$c_0=-1, c_1=2, c_2=98$.
\end{remark}

Proposition~\ref{prp1lem} and formula~\eqref{IZformula} imply immediately 
the following corollary.
\begin{cor}\label{cor3}
For any fixed $g\geq 0$, the following asymptotic formula is true:
\beq\label{largen}
{\rm Vol} \, \Q_{g,n} \;\sim\;  \kappa_g \, \frac{n^{\frac{g}2} \pi^{2n}}{2^n} \quad (n\rightarrow \infty) \,,
\eeq
where 
\beq\label{kappag}
 \kappa_g \= \frac{64\,\pi^{6g-\frac{11}2}}{384^g\,\Gamma(\frac{5g-1}2)} \, c_g \,,  
\eeq
and $c_g$ are defined by \eqref{cgdef}.
\end{cor}
\noindent The reader may notice that certain universality 
found in~\cite{DYZ} about asymptotics of enumerations 
related to~$\overline{\mathcal{M}}_{g,n}$ 
reappears in~\eqref{largen}, \eqref{kappag}. 
The first few~$\kappa_g$ are given by 
$\kappa_0 = 32/\pi^6$, $\kappa_1 = \pi^{\frac12}/3$,  $\kappa_2 = 7 \pi^6/1080$, 
$\kappa_3 = 245 \pi^{25/2}/7962624$.

\section{Relations for the MV volumes}\label{section3}
The goal of this section is to prove Theorem~\ref{thmequation} and Corollary~\ref{Kazdeduced}. 

\begin{proof}[Proof of Theorem~\ref{thmequation}]
It was shown by Buryak~\cite{Buryak} that the Hodge hierarchy associated with~$\Lambda(s)$ 
 is normal Miura equivalent \cite{DZ-norm,DLYZ1} 
to the intermediate long wave (ILW) hierarchy. 
To be precise, define~$\tilde u=\tilde u(\bt;s;\e)$ by 
\beq\label{Miuratimes}
\tilde u(\bt;s;\e) \:= \sum_{g=0}^\infty \e^{2g}  \frac{(-1)^g s^g}{2^{2g} (2g+1)! } \frac{\p^{2g} u}{\p t_0^{2g}},
\eeq
where $u$ is defined in~\eqref{defineu} with the specialization $\bs=\bs^*$; then $\tilde u$
satisfies~\cite{Buryak} the ILW hierarchy, which has the first two flows
\begin{align}
& \tilde u_{t_1} \= \tilde u \frac{\p \tilde u}{\p t_0}  
\+ \sum_{g\geq1} \frac{|B_{2g}|}{(2g)!} \e^{2g} s^{g-1} \frac{\p^{2g+1} \tilde u}{\p t_0^{2g+1}}\,, \label{ilw1}\\
& \tilde u_{t_2}\= \frac 1 2 \tilde u^2 \frac{\p \tilde u}{\p t_0} \+ \sum_{g\geq 1} \frac{|B_{2g}|}{(2g)!} \e^{2g} 
\frac{s^{g-1}}{4}\biggl(2 \p_{t_0} \Bigl(\tilde u\frac{\p^{2g} \tilde u}{\p t_0^{2g}}\Bigr) + \frac{\p^{2g+1} (\tilde u^2)}{\p t_0^{2g+1}} \biggr) \nn\\
& \qquad \qquad \+ \sum_{g\geq 2} \frac{|B_{2g}|}{(2g)!} (g+1) \e^{2g} s^{g-2} \frac{\p^{2g+1} \tilde u}{\p t_0^{2g+1}} \, .\label{ilw2-zh}
\end{align}

Let us now do the specialization~\eqref{specialvalue} with $s=1$, and denote the series $u(\bt^*;\bs^*;\e)|_{s=1}$, $\tilde u(\bt^*;s;\e)|_{s=1}$ by
$u=u(x,\e)$, $\tilde u=\tilde u(x,\e)$, 
respectively. Then $u(x,\e) =\e^2\p_x^2(\H(x,\e))$, and 
from~\eqref{Miuratimes} it follows that 
$\tilde u(x,\e)$ and $u(x,\e)$ are related by
\beq\label{utu}
 \tilde u \= \sum_{g=0}^\infty \e^{2g}  \frac{(-1)^g}{2^{2g} (2g+1)! } \frac{\p^{2g} u}{\p x^{2g}} \,.
\eeq

\begin{prop} \label{BuryakILW}
The series~$\tilde u = \tilde u(x,\e)$ satisfies the following non-linear equation:
\beq\label{EL}
x \+ \frac{\tilde u^2}{2} \+ \sum_{g=1}^\infty \e^{2g} \frac{|B_{2g}|}{(2g)!} \frac{\p^{2g} \tilde u}{\p x^{2g}} \= \tilde u \,.
\eeq
\end{prop}
\begin{proof}
Recall that the Hodge partition function $Z=Z(\bt;\bs;\e):=e^{\H(\bt;\bs;\e)}$ satisfies the 
string equation (cf.~e.g.~\cite{DLYZ1,DuY}), that is,
\beq\label{stringeqZ}
\sum_{i=0} t_{i+1} \frac{\p Z}{\p t_i} \+ \frac{t_0^2}{2\epsilon^2} Z \+ \frac{s_1}{24} Z \= \frac{\p Z}{\p t_0} \,. 
\eeq
Dividing both sides of~\eqref{stringeqZ} by~$Z$ 
and differentiating with respect to~$x$ we obtain 
\beq\label{stringOmega}
\sum_{i=0}^\infty t_{i+1} \frac{\p^2 \H(\bt;\bs;\e)}{\p t_i\p x} \+ \frac{x}{\epsilon^2}  \= \frac{\p^2 \H(\bt;\bs;\e)}{\p x^2} \,. 
\eeq
We recall that 
\beq\label{twopointfunctions}
\e^2\frac{\p^2 \H(\bt;\bs;\e)}{\p t_i\p x} \= \Omega_{i,0}\bigl(u(\bt;\bs;\e),u_x(\bt;\bs;\e),\dots\bigr)\,,\quad i\geq0\,, 
\eeq
where $\Omega_{i,0}$ are certain differential polynomials~\cite{Buryak,DLYZ1} of~$u$. 
Then by using the Miura transformation~\eqref{Miuratimes} we obtain 
\beq\label{stringOmegatilde}
\sum_{i=0}^\infty t_{i+1} \widetilde \Omega_{i,0}\bigl(\tilde u(\bt;\bs;\e), \tilde u_x(\bt;\bs;\e),\dots\bigr) \+ x  
\= \tilde u(\bt;\bs;\e)\,. 
\eeq
Here $\widetilde\Omega_{i,0}$, $i\geq 0$ are differential polynomials of~$\tilde u$. 
Buryak~\cite{Buryak} showed that the Miura transformation~\eqref{Miuratimes} 
transforms the Hamiltonian structure~$P$
of the linear Hodge hierarchy to $\p_x$, in particular, $\tilde u(\bt;\bs;\e)$ satisfies the Hamiltonian system
\beq
 \frac{\p \tilde u}{\p t_1} \= \p_x \frac{\delta \tilde {\bar h}_1}{\delta \tilde u(x)}   \,,
\eeq 
where
\[ \tilde {\bar h}_1 \=  \int \biggl(\frac{\tilde{u}^3}{6} \+ \sum_{g=1}^\infty \frac{|B_{2g}|}{2(2g)!} \tilde u \tilde u_{2g} \biggr) dx\,. \]
Therefore, according to~\cite{DLYZ1} we know that 
\beq\label{Omega10}
\widetilde \Omega_{1,0} \= \frac{\delta \tilde {\bar h}_1}{\delta \tilde u(x)} \=  \frac{\tilde u^2}{2} \+ \sum_{g=1}^\infty \e^{2g} \frac{|B_{2g}|}{(2g)!} \frac{\p^{2g} \tilde u}{\p x^{2g}}  \,.
\eeq
Thus equation~\eqref{stringOmegatilde} with the specialization $s=1$
leads to~\eqref{EL}. The proposition is proved.
\end{proof}

We are in a position of proving equation~\eqref{theequation}.
Indeed, observe that 
\beq
\sum_{g\geq 1} \e^{2g} \frac{|B_{2g}|}{(2g)!} \p_x^{2g}  \= 1 \,-\,\frac{ i}2 \e \, \p_x  \,-\, \frac{i\e\p_x}{e^{i \e \p_x}-1},
\eeq
so it follows from ~\eqref{EL} that
\beq
x \+ \frac{\tilde u^2}{2} \,-\,\frac{ i}2 \e \, \p_x (\tilde u)  \,-\, \frac{i\e\p_x}{e^{i \e \p_x}-1} (\tilde u)  \= 0 \,.
\eeq
By using the fact that $\tilde u = -i \e \p_x (\H_+-\H_-)$ we arrive at equation~\eqref{theequation}. 

We will now prove equation~\eqref{anotherequation}.
We first switch on the $t_2$-dependence and denote it by~$t$ 
in the specialization~\eqref{specialvalue}. 
More precisely, we consider 
\begin{align}
& \H \= \H(x,t,\e) \:=  \sum_{g,n\geq 0} 
\sum_{j=0}^g \frac{\langle\lambda_j \tau_0^n \tau_2^{3g-3+n-j}\rangle}{(3g-3+n-j)!} \, \e^{2g-2} \frac{x^n}{n!} t^{3g-3+n-j} \, ,
\end{align}
and denote $\H_{\pm} := \H\bigl(x\pm\frac{i \e}{2},t,\e\bigr)$. 
Then by using equation~\eqref{ilw1} and an argument like the one we used above to derive  
equation~\eqref{theequation}, we find that~$\H$ satisfies the following equation:
\begin{align}
&t\frac{\e^2}2 \Bigl[\p_x(\H_+\,-\, \H_-)\Bigr]^2 \+ t\frac{\e^2}2 \p_x^2 \bigl( \H_+ + \H_-\bigr) \,-\, (1-t) \, i \,\e \, \p_x(\H_+\,-\, \H_-) \= x \,.\label{H21} 
\end{align}
Then by using equations~\eqref{ilw2-zh} and~\eqref{H21} we obtain the following equation for~$\H$:
\begin{align}
& -i \,\e\, \p_t \, (\H_+\,-\, \H_-) \= \frac16 \tu^3 + \frac34 \tu^2 + \tu - \frac{i\e}2 \tu \tu_x -\frac{3i\e}4 \tu_x - \frac{\e^2}6 \tu_{xx} 
 \+ \frac{x}{2t} - \frac{i\e}{4t}  \nn\\
& \qquad \qquad \qquad \qquad \qquad \quad - \frac{1+2t}{2t} \e^2 \p_x^2(\H_-) + \frac{i\e^3}{4} \p_x^3(\H_-) - \frac{\e^2}2  \tu \, \p_x^2(\H_-) \,. \label{H22}
\end{align}
Here we recall that $\tilde u=-i\e\p_x(\H_+\,-\, \H_-)$, and we also used Theorem~A to get the constant in~$x$ term $- i\e/4t $.
It is not difficult to deduce from Theorem~A the following homogeneity property for~$\H$:
\beq\label{homogenH}
t \frac{\p \H}{\p t} \+ \Bigl(x-\frac1t\Bigr) \frac{\p \H}{\p x} \+ \e \frac{\p \H}{\p \e} \= - \frac1{24} - \frac1{24 t} -\frac{x^2}{2 \e^2 t} \,.
\eeq
From the above equations \eqref{H21}--\eqref{homogenH} we arrive at equation~\eqref{anotherequation}. 
The theorem is proved.
\end{proof}

Let us proceed to prove Corollary~\ref{Kazdeduced}.
\begin{proof}[Proof of Corollary~\ref{Kazdeduced}] 
Differentiating equation~\eqref{anotherequation} with respect to~$x$ we obtain 
\[
\biggl(\e \p_x \p_\e  + \frac12 \p_x + \frac12x\p_x^2 - \frac{\e^2}{24} \p_x^4\biggr)  \bigl(\H_{+} - \H_{-}\bigr)    \+ 
\frac{\e^2}{4} \bigl[\p_x(\H_{+}-\H_{-})\bigr]^2  \bigl[\p_{x}^2(\H_{+}-\H_{-})\bigr]  \=  0 \,,\]
so from equation~\eqref{theequation} it follows that 
\[
\biggl(\e \p_\e  + \frac12  + x\p_x - \frac{\e^2}{24} \p_x^3\biggr) \circ \p_x \bigl(\H_{+} - \H_{-}\bigr)    \,-\,
\frac{\e^2}{4} \Bigl[ \bigl(\p_x^2 ( \H_+)\bigr)^2 - \bigl(\p_{x}^2(\H_{-})\bigr)^2\Bigr]  \=  0   \,. \]
Observing that $\bigl[ x  \p_x,e^{\pm i\e\p_x/2}\bigr]= \mp \frac{i\e}2 e^{\pm i\e\p_x/2} \p_x $ one can simplify this equation and find
\beq\label{equiv1kaz}
\Bigl(e^{\frac{i\e\p_x}2}-e^{-\frac{i\e\p_x}2}\Bigr)  
\Biggl[\biggl(\e \p_\e  + \frac12  + x\p_x - \frac{\e^2}{24} \p_x^3\biggr) \circ \p_x \bigl(\H \bigr)    \,-\,
\frac{\e^2}{4} \bigl[\p_x^2 ( \H )\bigr]^2 \Biggr]  \=  0   \,. 
\eeq
Since the operator $(e^{\frac{i\e\p_x}2}-e^{-\frac{i\e\p_x}2}) / \p_x$ 
is invertible on power series of~$x$, 
we find that equation~\eqref{equiv1kaz} is equivalent to
\beq\label{equiv2kaz}
\p_x  \Biggl[\biggl(\e \p_\e  + \frac12  + x\p_x - \frac{\e^2}{24} \p_x^3\biggr) \circ \p_x \bigl(\H \bigr)    \,-\,
\frac{\e^2}{4} \bigl[\p_x^2 ( \H )\bigr]^2 \Biggr]  \=  0   \,. 
\eeq
It follows that 
\[\biggl(\e \p_\e  + \frac12  + x\p_x - \frac{\e^2}{24} \p_x^3\biggr) \circ \p_x \bigl(\H \bigr)    \,-\,
\frac{\e^2}{4} \bigl[\p_x^2 ( \H )\bigr]^2 \= C(\e)\,,\]
where $C(\e)=\sum_{g\geq0} \e^{2g-2} C_g$ with $C_g$ being constants.  It remains to show that $C_g$ all vanish. Indeed, 
for $g=0$ and $g=1$, this can be verified directly with the explicit expressions of~$\H_0$ and~$\H_1$ given in Lemma~\ref{Txstructure}.
For $g\geq2$, by using Lemma~\ref{Txstructure} and the fact that $\p_x = -\frac{1}{T}\p_T$ we arrive at $C_g=0$.
The  corollary is proved.
\end{proof}

Let us now show that Corollary~\ref{Kazdeduced} implies Kazarian's recursion 
on the linear Hodge integrals    
  $\frac{(5g-3-j)(5g-5-j)}{(3g-3-j)!}\int_{\overline{\mathcal{M}}_{g,3g-3-j}} 
\lambda_j  \psi_{1}^2 \cdots \psi_{3g-3-j}^2$.
Indeed, differentiating~\eqref{deducedfrom} with respect to~$x$ we find that 
the series~$u=\e^2 \p_x^2(\H)$ satisfies the equation
\beq\label{equationforumay1}
2\e u_{\e} \+ 2x u_x \,-\, u  \=  \p_x \Bigl(\frac12 u^2\Bigr)\+\frac1{12}\e^2 u_{xxx} \,.
\eeq
Denote 
\beq\label{genusexpansionsuandutilde}
 u(x,\e) =: \; \sum_{g\geq 0}\e^{2g} u^{[g]}(x)\,.
 \eeq
Then we can write~\eqref{equationforumay1} equivalently as follows:
\beq\label{newrecursionforu}
\Bigl(4g-1+2x\frac{d}{dx}\Bigr) \bigl(u^{[g]}\bigr) \=  \frac12 \frac{d}{dx} \biggl(\sum_{g_1+g_2=g} u^{[g_1]}u^{[g_2]}\biggr)
 \+ \frac1{12} \frac{d^3}{dx^3} \bigl(u^{[g-1]}\bigr) \,,\quad g\geq 0\,.
\eeq
To proceed we note that it follows easily from Lemma~\ref{Txstructure} that $u^{[g]}(x)$ has the expression 
\begin{align}
& u^{[0]} = 1 - T \,, \quad  
u^{[1]} \= \frac{1}{12} \frac{1}{T^{4}} \+ \frac{1}{24} \frac{1}{T^{3}}  \,, \label{u0u1}\\
& u^{[g]} \= \sum_{j=0}^{g} \frac{\langle \lambda_j \tau_2^{3g-3-j} \rangle_g}{(3g-3-j)!} 
\frac{\prod_{i=0}^{1}(5g-5-j+2i)}{T^{5g-1-j}} \,, \quad g\geq 2\,. \label{expandug}
\end{align}
Thus using the 
fact that $\frac{d}{dx}= -\frac1T \frac{d}{dT}=:D_T$ we find that~\eqref{newrecursionforu} is equivalent to 
\beq\label{dtuu}
\bigl(4g-1+(1-T^2)D_T\bigr) \bigl(u^{[g]}\bigr) \=  \frac12 D_T \biggl(\sum_{g_1,g_2\geq 0 \atop g_1+g_2=g} u^{[g_1]}u^{[g_2]}\biggr) \+ \frac1{12} D_T^3 \bigl(u^{[g-1]}\bigr)\,.
\eeq
Substituting~\eqref{u0u1}, \eqref{expandug} into~\eqref{dtuu} we find 
\begin{align}
c_{g,j} \= & \frac{g+1-k}{5g-2-j}\, c_{g,j-1} \+ \frac{(5g-6-j)(5g-4-j)}{12} \, c_{g-1,j}\notag \\ 
& \+ \frac12 \sum_{g_1,g_2\geq 1, j_1,j_2\geq 0 \atop g_1+g_2=g,  j_1+j_2=j} c_{g_1,j_1} c_{g_2,j_2} \,,\qquad g\geq1,\, 0\leq j\leq g\,,\label{kaz}
\end{align}
where the numbers $c_{g,j}$ are defined by 
\beq
c_{g,j} \:=  \frac{\langle \lambda_j \tau_2^{3g-3-j} \rangle_g}{(3g-3-j)!}  \prod_{i=0}^{1}(5g-5-j+2i) \,.
\eeq
The recursion relations~\eqref{kaz} for $c_{g,j}$ were obtained by Kazarian~\cite{Kazarian} 
from the KP hierarchy~\cite{Kazarian1} satisfied by the linear Hodge integrals.

It is not clear at the moment whether Corollary~\ref{Kazdeduced} and Lemma~\ref{Txstructure} imply Theorem~\ref{thmequation}.

We end this section with two remarks on the computational aspects.
Firstly, as a consequence of equation~\eqref{theequation} and Lemma~\ref{Txstructure}, 
the~$u^{[g]}$ can be computed from the recursion 
\begin{align}
& u^{[0]} \= 1-T \,, \nn\\
& u^{[g]} \= \frac1{2T} \sum_{0\leq g_1,g_2\leq g-1 \atop g_1+g_2+j_1+j_2=g} \Bigl(-\frac14\Bigr)^{j_1+j_2}
\frac{D_T^{2j_1}\bigl(u^{[g_1]}\bigr) D_T^{2j_2}\bigl(u^{[g_2]}\bigr)}{(2j_1+1)!(2j_2+1)!} 
 \,-\, \frac1T \sum_{j=1}^g \Bigl(-\frac14\Bigr)^{j} \frac{D_T^{2j}\bigl(u^{[g-j]}\bigr)}{(2j)!}  \,, \nn
\end{align}
where $g\geq 1$. 
Then one can further compute $\H_g$, $g\geq 2$ from $u^{[g]}$ via
\beq
\HHod_g \= \sum_{j=0}^{g}  \frac{C_{g,j}}{T^{5g-5-j}} \,, \qquad 
C_{g,j} \= \frac{{\rm coefficient~of}~1/T^{5g-1-j}~{\rm in}~ u^{[g]}}{(5g-3-j)(5g-5-j)}  \quad (0\leq j\leq g)\,.
\eeq
Secondly, the series $\tilde u$ (see~\eqref{utu}) also presents good properties. Denote 
\beq\label{genusexpansionsutilde} 
\tilde u(x,\e) =: \; \sum_{g\geq 0}\e^{2g} \tilde u^{[g]}(x) \,. 
\eeq
If then follows from \eqref{u0u1}, \eqref{expandug}
that $\tilde u^{[g]}$ has the expression
\begin{align}
& \tilde u^{[0]} = 1 - T \,, \quad \tilde u^{[1]} \= \frac{1}{12 T^4} \,, \quad \tilde u^{[g]} \=  \sum_{j=0}^g  \frac{d_{g,j}}{T^{5g-1-j}} ~ (g\geq2)\,, \label{expandtu}
\end{align}
where $d_{g,j}\in \QQ$ are constants. In terms of intersection numbers we have for $g\ge 2$,
\begin{align}
& \tilde u^{[g]} \= 
\sum_{g_1=0}^{g-2}  \frac{(-1)^{g_1}}{2^{2g_1} (2g_1+1)! }  \sum_{j=0}^{g-g_1} \frac{\langle \lambda_j \tau_2^{3g-3g_1-3-j} \rangle_{g-g_1}}{(3g-3g_1-3-j)!} \frac{\prod_{i=0}^{1+2g_1}(5g-5g_1-5-j+2i)}{T^{5g-g_1-1-j}} \nn\\
&\qquad \quad   \+\frac{(-1)^{g-1}2^{2g}}{12}   \frac1{T^{4g}} \+ \frac{(-1)^{g} (4g-3)!!}{2^{2g} (2g+1)! }  \, \frac{5-2g}{6}\, \frac{1}{T^{4g-1}} \,.\nn
\end{align}
Substituting~\eqref{genusexpansionsutilde} 
into~\eqref{EL} we find that  $\tilde u^{[g]}$, $g \geq 0$ satisfy the following recursion
\begin{align}
& \tilde u^{[0]} \= 1 - T \,, \label{polyalg1}\\
& \tilde u^{[g]} \= \frac1{2T} \sum_{g_1=1}^{g-1}  \tilde u^{[g_1]}\tilde u^{[g-g_1]} \+ 
\frac1{T}\sum_{g_1=1}^{g} \frac{|B_{2g_1}|}{(2g_1)!} D_T^{2g_1} \bigl(\tilde u^{[g-g_1]}\bigr) \,, \quad g\geq 1\,. \label{polyalg2}
\end{align}
This recursion gives 
an algorithm for computing~$\tilde u$. 
From~\eqref{utu} we know that 
\[
u \=  \tilde u \+ \sum_{g\geq 1} \e^{2g} \frac{2^{2g-1}-1}{2^{2g-1}}  \frac{|B_{2g}|}{(2g)!}  D_T^{2g} (\tilde u) \,.
\]
Therefore, for $g\geq 0$, 
\[u^{[g]} \=\tilde u^{[g]} \+ \sum_{g_1=1}^g \frac{2^{2g_1-1}-1}{2^{2g_1-1}} \frac{|B_{2g_1}|}{(2g_1)!} 
D_T^{2g_1} \bigl( \tilde u^{[g-g_1]}\bigr) \,. \]
So this gives rise to another algorithm for computing the MV volumes. 
One could also use~\eqref{anotherequation} to study~$\tilde u$.

\section{Asymptotics of the area Siegel--Veech constants}\label{section4}
In this section we use Goujard's formula to compute the area Siegel--Veech (SV) constants 
associated with principal strata of moduli spaces of quadratic differetials. 
Indeed, according to Goujard~\cite{Goujard} the area 
SV constants can be expressed explicitly in terms of 
the number~$a_{g,n}$ as follows:
\beq\label{Goujardformula}
C_{\rm area} (\mathcal{Q}_{g,n}) \= \frac{\pi^{-2}}{4a_{g,n}} 
\Biggl(n(n-1)a_{g,n-1} \+  a_{g-1,n+2}  \; + \!\! \sum_{g_1,g_2\ge 0, \, n_1,n_2 \ge 1 \atop {g_1+g_2=g, \, n_1+n_2=n+2\atop 3g_i-3+n_i > 0 \, (i=1,2) }} 
\binom{n}{n_1-1} a_{g_1,n_1}a_{g_2,n_2} 
\Biggr) \,.
\eeq 
The result in this section is a refinement 
of the conjectural formula for the large~$g$ asymptotics 
of $C_{\rm area} (\mathcal{Q}_{g,n})$ given in~\cite{DGZZ,ADGZZ} to the following more precise asymptotic statement.

\begin{conj}\label{conjecture2}
For any fixed $n\geq 0$, we have the asymptotic formula
\beq
C_{\rm area} (\mathcal{Q}_{g,n}) \;\sim\;  \sum_{k=0}^\infty \frac{C_k(n)}{g^k} \,,\qquad g\to\infty \,,
\label{asvclargeg} 
\eeq
where each $C_k(n)$ is a polynomial with rational coefficients in~$n$ and $M=-\pi^2/144$, 
with the first four of them being 
\begin{align}
C_0(n) &\= \frac14\,, \quad C_1 (n) \= \frac1{48}n^2\,-\, \frac3{16} n \+ \frac{1-2 M}{4} \,, \nn\\
C_2 (n) & \= - \frac{5+12 M}{576} \,n^3 
\+ \frac{59+180 M}{576}\,n^2 
 \,-\, \frac{11+6 M-72 M^2}{32}  \,n 
 \+ \frac{23+15 M-648 M^2}{72} \,,  \nn \\ 
C_3 (n) &\= 
\frac{4+17M+54 M^2}{1152}\, n^4 \,-\, 
\frac{179+978 M+3564 M^2}{3456} \, n^3 \nn\\
 & \qquad \+ \frac{929+5169 M+13554 M^2-42768 M^3}{3456} \, n^2  \nn\\ 
 & \qquad \,-\, \frac{989+4851 M -4428 M^2 - 192456 M^3}{1728} \,  n \nn\\
 & \qquad \+  
 \frac{295+1165 M-16140 M^2-105300 M^3+253692 M^4}{720} \,.\nn 
\end{align}
\end{conj}

The asymptotic formula~\eqref{asvclargeg} with $\sum_{k=0}^\infty C_k(n)/g^k$ replaced by~$1/4$ 
becomes the ADGZZ conjecture for the area SV constants.  As we mentioned in the Introduction, the 
above Conjecture~\ref{conjecture2} is also not based on theoretical reasoning but on numerical computations. 
Very recently Aggarwal~\cite{A3} proved the ADGZZ conjecture for the area SV constants by showing that   
the leading term asymptotics in~\eqref{mvlargeg} implies the leading term asymptotics in~\eqref{asvclargeg} 
with the knowledge of Goujard's formula~\eqref{Goujardformula}. 
However, we do not know whether
Conjecture~\ref{conjasymg} implies Conjecture~\ref{conjecture2} in the same way. This would be an interesting
point to investigate next.

\medskip
\medskip

\noindent Di Yang

\noindent School of Mathematical Sciences, University of Science and Technology of China\\
Hefei 230026, P.R.~China \\
diyang@ustc.edu.cn

\medskip

\noindent Don Zagier

\noindent Max-Planck-Institut f\"ur Mathematik, Bonn 53111, Germany \\
and International Centre for Theoretical Physics, Trieste 34014, Italy \\
\noindent dbz@mpim-bonn.mpg.de

\medskip

\noindent Youjin Zhang

\noindent Department of Mathematical Sciences, Tsinghua University \\ 
Beijing 100084, P.R.~China\\
youjin@mail.tsinghua.edu.cn

\end{document}